\documentclass[Times1COL]{WileyNJDv5} 

\usepackage{moreverb,url}

\usepackage{mathrsfs,bm}
\usepackage{dsfont}
\usepackage{float}
\usepackage{bbm}
\usepackage[utf8]{inputenc}
\usepackage{amsmath}
\usepackage{graphicx}
\usepackage{amssymb}
\usepackage{amsthm}
\usepackage{verbatim}
\usepackage{epsfig}
\usepackage[labelfont=bf]{caption}
\usepackage{enumerate}
\usepackage{circuitikz}
\usetikzlibrary{calc} 
\usepackage{subfig}
\usepackage{epstopdf}
\usepackage{needspace}
\usepackage{hyperref} 

\makeatletter
\newcommand{\pushright}[1]{\ifmeasuring@#1\else\omit\hfill$\displaystyle#1$\fi\ignorespaces}
\renewcommand\@biblabel[1]{#1.\ }
\makeatother

\def\begquo{\begin{quote}}
	\def\endquo{\end{quote}}
\def\begequarr{\begin{eqnarray}}
	\def\endequarr{\end{eqnarray}}
\def\begequarrs{\begin{eqnarray*}}
	\def\endequarrs{\end{eqnarray*}}
\def\begarr{\begin{array}}
	\def\endarr{\end{array}}
\def\begequ{\begin{equation}}
	\def\endequ{\end{equation}}
\def\lab{\label}
\def\begdes{\begin{description}}
	\def\enddes{\end{description}}
\def\begenu{\begin{enumerate}}
	\def\begite{\begin{itemize}}
		\def\endite{\end{itemize}}
	\def\endenu{\end{enumerate}}

\def\lef[{\left[\begin{array}}
	\def\rig]{\end{array}\right]}

\def\begcen{\begin{center}}
	\def\endcen{\end{center}}
\def\begrem{\begin{remark}\rm}
	\def\endrem{\end{remark}}
\def\begdef{\begin{definition}}
	\def\enddef{\end{definition}}
\def\begpro{\begin{proposition}}
	\def\endpro{\end{proposition}}
\def\begfac{\begin{fact}}
	\def\endfac{\end{fact}}
\def\begass{\begin{assumption}}
	\def\endass{\end{assumption}}

\def\begsubequ{\begin{subequations}}
	\def\endsubequ{\end{subequations}}

\def\begmat#1{\begin{bmatrix}#1\end{bmatrix}}
\def\begali#1{\begin{align}{#1}\end{align}}
\def\begalis#1{\begin{align*}{#1}\end{align*}}

\def\cale{{\cal E}}

\def\cals{{\cal S}}
\def\cale{{\cal E}}

\def\cald{{\cal D}}

\def\L2e{{\cal L}_{2e}}

\def\rea{\mathds{R}}

\def\adj{\mbox{adj}}
\def\col{\mbox{col}}
\def\hal{{1 \over 2}}

\def\min{{\mbox{min}}}

\def\IJACSP{{\it Int. J. on Adaptive Control and Signal Processing}}

\def\TAC{{\it IEEE Trans. Automatic Control}}

\def\TPE{{\it IEEE Trans. Power Electronics}}

\def\SCL{{\it Systems and Control Letters}}

\def\CST{{\it IEEE Trans. Control Systems Technology}}

\def\CSM{{\it IEEE Control Systems Magazine}}

\def\cale{{\cal E}}

\def\hal{{1 \over 2}}

\usepackage{color}

\usepackage[prependcaption,colorinlistoftodos]{todonotes}

\newcommand{\aTerm}{$a_1a_2 - a_0$}

\articletype{Article Type}%

\received{Date Month Year}
\revised{Date Month Year}
\accepted{Date Month Year}
\journal{Journal}
\volume{00}
\copyyear{2023}
\startpage{1}

\begin{document}

\title{Voltage Control of the Boost Converter: PI vs. Nonlinear Passivity-based Control}

\author[1,2]{Leyan Fang}

\author[2]{Romeo Ortega}

\author[3]{Robert Griñó}

\authormark{TAYLOR \textsc{et al.}}
\titlemark{PLEASE INSERT YOUR ARTICLE TITLE HERE}

\address[1]{\orgdiv{Center for Control Theory and Guidance Technology}, \orgname{Harbin Institute of
Technology}, \orgaddress{\state{Harbin}, \country{China}}}

\address[2]{\orgdiv{Departamento Acad\'{e}mico de Ingenier\'ia El\'ectrica y Electr\'onica}, \orgname{ITAM}, \orgaddress{\state{Ciudad de M\'exico}, \country{M\'exico}}}

\address[3]{\orgdiv{Institute of Industrial and Control Engineering}, \orgname{Universitat Politècnica de Catalunya}, \orgaddress{\state{Barcelona}, \country{Spain}}}

\corres{ Corresponding author Leyan Fang, Departamento Acad\'{e}mico de Ingenier\'ia El\'ectrica y Electr\'onica, ITAM, R\'io Hondo 1, 01080,  Ciudad de M\'exico, M\'{e}xico. \email{leyan.fang@itam.mx}}

\abstract[Abstract]{We carry-out a detailed analysis of {\em direct voltage control} of a Boost converter feeding a simple resistive load. First, we prove that using a classical PI control to stabilize a desired equilibrium leads to a very complicated dynamic behavior consisting of two equilibrium points, one of them {\em always unstable} for all PI gains and circuit parameter values. Interestingly, the second equilibrium point may be rendered {\em stable}---but for all tuning gains  leading to an extremely large value of the circuit current and the controller integrator state. Moreover, if we neglect the resistive effect of the inductor, there is only one equilibrium and it is {\em always unstable}. From a practical point of view, it is important to note that the only useful equilibrium point is that of minimum current and that, in addition, there is always a resistive component in the inductor either by its parasitic resistance or by the resistive component of the output impedance of the previous stage. In opposition to this troublesome scenario we recall three nonlinear voltage-feedback controllers, that ensure asymptotic stability of the desired equilibrium with simple gain tuning rules, an easily defined domain of attraction and smooth transient behavior. Two of them are  {\em very simple, nonlinear, static} voltage feedback rules, while the third one is a variation of the PID scheme called PID-Passivity-based Control (PBC). In its original formulation PID-PBC requires full state measurement, but we present a modified version that incorporates a {\em current observer}. All three nonlinear controllers are designed following the principles of PBC, which has had enormous success in many engineering applications. }


\maketitle
\renewcommand{\thefootnote}{\arabic{footnote}}

\section{Introduction}\label{sec1}

The ever increasing demand for smaller size, portable, and lighter weight high performance DC-DC power converters for industrial, communications, residential,
and aerospace applications, has made the  control of these devices a topic of widespread interest, see \cite{ROSetal} for a recent survey. The {most} popular configuration for this kind of power converters is the so-called Boost converter, which provides low voltage amplification and current ratings for loads at constant switching frequency. Various attempts have been made to formulate suitable control strategies for this switching device, because its overall dynamic performance is largely determined by the controller. The task of designing a controller is complicated by the fact that the mathematical model of the Boost converter is a {\em bilinear} system, and its linear approximation yields very poor approximations to its behavior. An additional difficulty in the control of Boost converters is that the output signal, that is the voltage fed to the load, has the behavior of a {\em non-minimum phase} system \cite{ISIbook}. Due to this behavior the vast majority of the practical controllers designed for the Boost converter are of the {{\em indirect} type,} where a PI loop is placed around the current, whose reference is determined to match the desired voltage value. This approach is clearly extremely fragile as the derivation of the current reference requires the exact knowledge of the system parameters or the design of an outer control loop for the converter output voltage, which increases the complexity of the whole controller. This issue has been extensively discussed in the literature, see {\em e.g.} \cite[Section 4.3.A]{ORTetalbookelsys} and references therein.

In this article we analyze in detail the dynamic behavior of the classical boost converter in closed-loop with {\em voltage-fed} control. It is shown in the paper that there are two important factors that significantly affect the control task. First, the aforementioned non-minimum phase behavior. Second the presence (or absence) of a parasitic resistor in the input inductor. In spite of all the research carried out to study the non-minimum phase behavior of the Boost converter, to the best of our knowledge, a detailed discussion of the zero dynamics, without appealing to system linearization, has not been carried out, in particular for the case when the aforementioned resistance is taken into account. \footnote{It should be noted that this resistive effect is due not only to the parasitic resistance in the inductor, but it appears also when we incorporate into the model the DC voltage source feeding the device---the so-called external series resistor.}  As we will see in the paper the control scenario dramatically changes when this resistor is taken into {account} in the model.

We concentrate our attention on two types of controllers wrapped around the voltage signal. First, the well-known PI controller that, as mentioned above, is rarely used in practical applications. A detailed study of this scenario provides a definite answer on when (and when note) this configuration can be used in practice. The second type of controllers is based on the principle of {\em passivity}---the so-called {\em passivity-based controllers} (PBC) \cite{ORTetalcsm,ORTetalaut02,SANVER,VANbook}. We study two recently developed simple, static nonlinear voltage-fed PBC, and vividly illustrate their excellent dynamic performance. Finally, we also consider a variation of PID, called PID-PBC, that has attracted a lot of interest in recent years \cite{ORTetalbookpid}.

The rest of the paper has the following structure. The mathematical model of the Boost converter and some of its control properties are presented in Section \ref{sec2}. The first main result of the paper---that is, the analysis of the voltage-fed PI control---is given in Section \ref{sec3}. In Section \ref{sec4} we analyze the PBCs and wrap-up the paper with some concluding remarks in Section \ref{sec5}. In Appendix \ref{appa} we present some cumbersome calculations needed to prove one of the main results.

\section{Boost Dynamic Model and Some Properties}\label{sec2}
In this section we provide some background material to study the dynamic behavior of the classical boost converter in closed-loop with {\em voltage-fed} controllers. First, to simplify the analytic expressions we start by proposing a suitable change of coordinates and time scaling. Second, we put into evidence the property of power balance of the converter. Thirdly, we compute the assignable equilibrium set \cite[Proposition B.1]{ORTetalbookpid} and, finally, we reveal the well-known fact that the {\em zero dynamics} of the converter, with respect to the voltage variable, is {\em non-minimum phase}.  
\subsection{The boost converter scaled model}
\lab{subsec21}
%
Consider the (average) model of the boost converter
\begin{subequations}\label{boost}
\begin{align}
	L\frac{ i_L}{dt} =& -R i_L- v_c u +E  \\
	C\frac{v_c} {dt} =& -Gv_C + i_L u  
\end{align}
\end{subequations}
where  $i_L(t) \in \rea$ is the inductor current, $v_C(t) \in \rea$ is the capacitor voltage and $u(t) \in \rea$ is the duty cycle. Moreover, $L>0$ and $C>0$ are the circuit inductance and capacitance, $R \geq 0$ is the inductor (parasitic) resistance, $G>0$ is the load conductance and $E>0$ is converter constant voltage source.\footnote{Equations (1) represent the dynamics of a bidirectional in current (and hence in power) boost converter or a unidirectional in current boost converter in its continuous conduction mode (CCM).}

\begin{fact}
	Under the following state variable change and time scaling 
		$$
		\begin{bmatrix}
			x_1\\x_2
		\end{bmatrix}= \begin{bmatrix} \frac{1}{E}\sqrt{\frac{L}{C}}& 0 \\ 0 & \frac{1}{E}\end{bmatrix}\begin{bmatrix} i_L\\v_c\end{bmatrix}, \tau =\frac{1}{\sqrt{LC}}t ,
	$$
	the system \eqref{boost} admits the standard $(f,g,h)$ representation 
	\begali{
	\nonumber
		\dot x& = \begmat{  -d_1 x_1+1\\ -d_2 x_2 } +\begmat{ -x_2\\x_1}u\\
		y&=x_2,
	\label{sys}
	}
	where $\dot x(\tau):= \frac{dx(\tau)}{d\tau }$, $d_1:=R\sqrt{\frac{C}{L}} \geq 0$ and $d_2:= G\sqrt{\frac{L}{C}}>0$.
\end{fact}

\begrem
\lab{rem1}
Notice that we are considering the possibility where the resistance $R=0$, which is a case that has often been studied. This is, however, not the case when the parasitic resistance of the inductor is considered, or the output impedance of the DC voltage source, or the previous stage, includes a resistive part.
\endrem 
\subsection{Power balance property}
\lab{subsec22}
%
Clearly, the system \eqref{boost} admits the port-Hamiltonian (pH) representation \cite{VANbook} 
	\begin{align}
	\label{phsys}
		\dot x =& F(u) \nabla H(x)+\begmat{ 1\\ 0 },
	\end{align}  
with total energy function $H(x)=\hal|x|^2$ and the interconnection and damping matrix
\begequ
\lab{f}
F(u):=\begmat{-d_1 & -u \\ u & -d_2}.
\endequ

The interest of the pH representation is that it puts in evidence the {\em power balance} property
$$
\underbrace{\dot H}_{\mbox{stored}}=\underbrace{-d_1 x_1^2-d_2 x_2^2}_{\mbox{dissipated}}+\underbrace{x_1 \times 1,}_{\mbox{supplied\; power}}
$$
which is a property satisfied by all physical systems. Notice that in the scaled model the supplied voltage $E$ is reduced to $1$---and the quantity $x_1 \times 1$ is the product of supplied voltage and extracted current. The current $x_1$ is the so-called ``natural output" $y=g^\top \nabla H$ \cite{VANbook}.

\begrem
\lab{rem2}
One important corollary stemming from the power balance equation is that the control action, similarly to the Coriolis forces in mechanical systems, ``doesn't do any work". Hence there are some stability properties of the system which are {\em independent} of the control action. 
\endrem 
\subsection{Assignable equilibrium set}
\lab{subsec23}
%
Now, applying \cite[Proposition B.1]{ORTetalbookpid}, we compute the {\em assignable equilibrium set} and the corresponding---uniquely defined---constant control value.

\begin{fact} 
\lab{fac2}
	Fix the desired equilibrium value of $x_2$ as $y_\star  >0$. If $d_1=0$ there is an {\em unique} assignable  equilibrium point given by
	\begequ
	\lab{assequ0}
	\cale_0=\{\bar x \in \rea^2\;\;|\;\bar x_1=d_2y_\star ^2,\;\bar x_2=y_\star \}, 
	\endequ
with the associated constant control value
$$
\bar u={1 \over y_\star }.
$$

If $d_1>0$ the corresponding assignable equilibrium set consists of {\em two} points given by
	\begin{align}
	\cale=\{\bar x \in \rea^2\;\;|\;\bar x_1={1 \over 2d_1}\Big(1 \pm \sqrt{1 -4d_1d_2y_\star ^2}\Big),\;\bar x_2=y_\star \}, 
	\label{assequ}
	\end{align}
with the associated constant control value
$$
\bar u={-1 \over \bar x_1^2+y_\star ^2}[(d_1 - d_2)\bar x_1-1]y_\star .
$$
\end{fact}

\begrem
\lab{rem3}
It is clear from \eqref{assequ} that, in order to ensure the existence of an equilibrium when $d_1 >0$ it is necessary and sufficient to satisfy the constraint
\begequ
\lab{equcon}
d_1d_2 < {1 \over 4y_\star ^2}.
\endequ
Notice that this constraint translates in the original coordinates into 
$$
RG < {E^2 \over 4(v_c^*)^2},
$$
where $v_c^*>0$ is the desired voltage value. We bring to the readers attention the fact that we rule-out the possibility of $d_1d_2 = {1 \over 4y_\star ^2}$, which would correspond to the current been equal to ${1 \over 2 d_1}$---which is not practically relevant. 
\endrem
\subsection{Zero dynamics of the boost converter with respect to the voltage}
\lab{subsec24}
%
It is well-known \cite{ORTetalbookelsys,SIRetal} that---when $d_1=0$---the {\em zero dynamics} with respect to the output voltage of the Boost converter,  is {\em non-minimum phase} \cite{ISIbook}. It is, precisely, this fact that motivates {most} practitioners to control the converter via the so-called {\em indirect method}, where a PI loop is placed around the current, whose reference is determined to match the desired voltage value. This approach is clearly extremely fragile as the derivation of the current reference requires the exact knowledge of the system parameters or the design of an outer control loop for the converter output voltage, which increases the complexity of the whole controller. This issue has been extensively discussed in the literature, see {\em e.g.} \cite[Section 4.3.A]{ORTetalbookelsys} and references therein. 

In spite of all this research, to the best of our knowledge, a detailed discussion of the zero dynamics of the Boost converter, without appealing to system linearization, has not been carried out, in particular for the case when $d_1>0$. To alleviate this deficiency is one of the motivations of this subsection. 

As is well-known, the zero-dynamics of a nonlinear system may be computed calculating its {\em normal form} \cite{ISIbook} or, alternatively, computing the differential equation satisfied by the system {\em input}, when its output is set to a constant equilibrium value---see \cite[Section 4.3.A]{ORTetalbookelsys} for the computation, via the second method, for the Boost converter. The lemma below exhibits its zero dynamics for the more general case when $d_1>0$. To the best of the authors' knowledge, this is the first time this result is reported.

\begin{lemma} 
\lab{lem1}
Consider the scaled version of the dynamics of the Boost converter \eqref{sys}. 
\begenu[{\bf F1}]
\item It's zero dynamics, with respect to the constant output $x_2=y_\star $ is given by
\begequ
\lab{zerdyn}
\dot u ={u \over d_2}\Big(u^2-{1 \over y_\star }u+d_1d_2\Big).
\endequ
\item This dynamics will have an equilibrium different from zero if and only if the condition \eqref{equcon} holds.
\item This equilibrium is {\em unstable}.
\endenu
\end{lemma}

\begin{proof}
The differential equation \eqref{zerdyn} is obtained differentiating $x_2$ twice and setting $x_2=y_\star $. 

Some simple calculations show that the term in parenthesis is equal to zero---that is, there is a second non-zero equilibrim---if and only if the condition \eqref{equcon} holds. To prove the instability of this equilibrium we study the polynomial
$$
w(u):=u^2-{1 \over y_\star }u+d_1d2,
$$  
which has a minimum at $u_\min={1 \over 2y_\star }$. It is easy to see that the non-zero equilibrium will be unstable if $w(u_\min)<0$, which is the case when \eqref{equcon} holds. In Fig \ref{fig1} we show the plot of the right hand side of \eqref{zerdyn} for two different cases.

 \begin{figure}
 	\begin{minipage}[h]{0.45\linewidth}
 		\center{\includegraphics[width=1\linewidth]{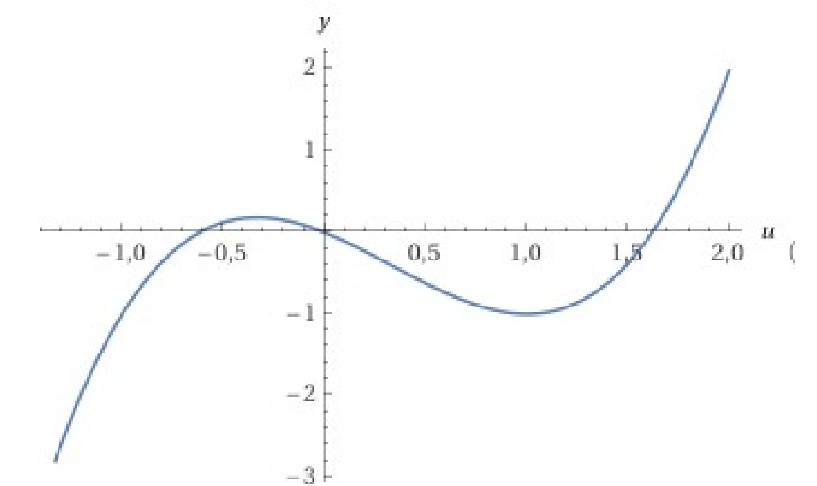}} a) \\
 	\end{minipage}
 	\hfill
 	\begin{minipage}[h]{0.45\linewidth}
 		\center{\includegraphics[width=1\linewidth]{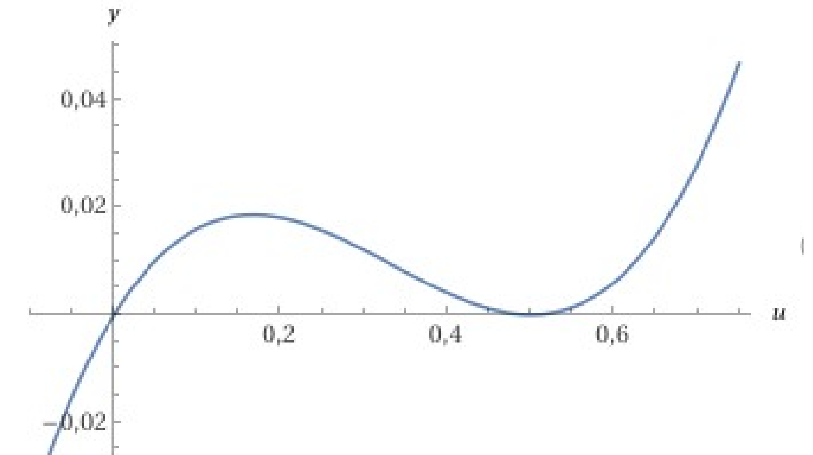}} \\b)
 	\end{minipage}
 	\vfill
 	\caption{Plot of the function $y(u)=u\Big(u^2-{1 \over y_\star }u+d_1d_2\Big)$: a) Case when $d_1d_2 < {1 \over 4y_\star ^2}$; b) Case when $d_1d_2 = {1 \over 4y_\star ^2}$.}
 	\label{fig1}
 \end{figure} 
 
\end{proof} 
\section{Main Stability Results for Voltage-fed PI Control}
\lab{sec3}
%
In this section we carry out the stability analysis of the PI control under the conditions of $d_1=0$ and $d_1>0$. This analysis consist of two parts, first the stability (in the Lyapunov sense) of their equilibria done invoking Lyapunov's Indirect Method \cite[Theorem 4.7]{KHAbook}. Second the determination of a domain of attraction for the case of stable equilibrium.
\subsection{PI control with $d_1=0$: instability of the equilibrium}
\lab{subsec31}
%
In this subsection we do the (local) stability analysis of the unstable equilibrium and present some simulations that illustrate this claim.
\subsubsection{Instability of the equilibrium}
\lab{subsubsec311}
%
\begpro
\lab{pro1}
Consider the system \eqref{sys} {\em with} $d_1=0$ in closed-loop with the {\em PI control}
\begin{equation}
\label{pi}
\begin{aligned}
\dot x_c &=y_\star  -x_2\\
u &= u_0 + K_I x_c+K_P(y_\star  -x_2),
\end{aligned}
\end{equation}
with the PI gains $K_P \geq 0$, $K_I > 0$ and $u_0 \in \rea$ a designer chosen parameter. The closed-loop system has a unique equilibrium point which is {\em unstable} for all settings of the tuning gains.
\endpro

\begin{proof}
 Define the state variables $\chi:=\col(x_1,x_2,x_c)$. The system \eqref{sys} with $d_1=0$ in closed-loop with the PI controller \eqref{pi} is of the form
\begequ
\lab{cloloo0}
\dot \chi=\begmat{1-u_0 \chi_2 -K_I \chi_2 \chi_3-K_P\chi_2y_\star +K_P \chi^2_2\\ -d_2 \chi_2 + u_0 \chi_1+K_I \chi_1 \chi_3+K_P\chi_1 y_\star -K_P\chi_1\chi_2\\-\chi_2+y_\star  }=:f(\chi).
\endequ
Invoking Lyapunov's Indirect Method \cite[Theorem 4.7]{KHAbook}, we will prove the instability claim of {Proposition} \ref{pro1}. Towards this end, we compute the Jacobian of the vector field $f(\chi)$ as
$$
\nabla f(\chi)=\begmat{0 & -u_0 -K_I \chi_3 -K_Py_\star +2K_P \chi_2 & -K_I \chi_2 \\ u_0 + K_I \chi_3 + K_P y_\star - K_P \chi_2  & -d_2 -K_P \chi_1 & K_I \chi_1 \\ 0 & -1 & 0}.
$$
Evaluating the Jacobian at the equilibrium $\bar \chi=\col(\bar \chi_1,y_\star,\bar \chi_3)$ yields
$$
\nabla f(\bar \chi)=\begmat{0 & -u_0  -K_I \bar \chi_3 + K_P y_\star  & -K_I y_\star  \\ u_0 + K_I \bar \chi_3& -d_2  -K_P \bar \chi_1& K_I  \bar \chi_1\\ 0 & -1 & 0}.
$$
The characteristic polynomial of this matrix is of the form $p(\lambda):=\lambda^3+a_2 \lambda^2+a_1 \lambda + a_0$, with 
\begalis{
& a_0:=- K_I(K_I \bar \chi_3 + u_0) y_\star\\
& a_1:= K_I^2\bar \chi_3^2+( 2u_0 \bar \chi_3 + \bar \chi_1-K_P \bar \chi_3 y_\star  )K_I  - u_0 y_\star K_P + u_0^2\\
& a_2:= K_P \bar \chi_1  + d_2.
}  
Based on the Routh-Hurwitz criterion, all roots of the characteristic polynomial will have negative real parts (i.e., system is stable) if and only if 
\begequ
\lab{stacon}
a_0>0,\;a_1>0,\;a_2>0,\;a_2 a_1>a_0.
\endequ
We now evaluate the equilibrium $\bar \chi_3$ as
$$
\bar \chi_3={1 \over K_I y_\star }(1-u_0 y_\star ),
$$
which upon replacement in $a_0$ yields $a_0=-K_I<0$---violating the first condition of \eqref{stacon}. Hence, the equilibrium is {\em unstable} for all circuit parameters and settings of the PI gains.
\end{proof}

\begrem
\lab{rem4}
This result proves that the claim of Proposition 3 (and consequently also the main result of Theorem 2) in \cite{ALVESP} is wrong.
\endrem

\begrem
\lab{rem5}
Notice that, setting $K_P=0$, we also prove that the pure integral control is also unstable for all values of $K_I>0$.
\endrem
\subsubsection{Simulation results}
\lab{subsubsec312}
%
In Fig. \ref{fig2} we present simulations of the closed-loop system \eqref{cloloo0} with different initial conditions for the following parameters:
$$
y_\star =2,d_2=1,K_I=1,K_P=2,u_0=\hal.
$$
From the figure it is clear that the equilibrium $\chi^u$ is, indeed, unstable.
 
It was observed that the surface 
$$
\cals_0:=\{(\chi_1,\chi_2)\;|\;\chi_1-d_2 \chi_2^2=0\},
$$
which clearly contains the assignable equilibrium point given in \eqref{assequ0}, defines a separatrix, in the sense that taking initial conditions ``below" the surface yields trajectories for $(\chi_1(t),\chi_2(t))$ converging to $(0,0)$. On the other hand, with initial conditions ``above" the surface the trajectory of $\chi_1(t)$ diverges to infinity. It is interesting to note that this surface precisely coincides with the place where the derivative of the systems energy $H(x)$ equal zero. Indeed, from \eqref{phsys} it is easy to see that, when $d_1=0$, we have
$$
\dot H=-d_2 x_2^2+x_1.
$$
This is consistent with the fact that ``below" the surface $\dot H<0$ and the energy decreases to zero. Unfortunately, no further analysis of this interesting behavior has been carried out.
  
  \begin{figure}
 	\begin{minipage}[h]{0.4\linewidth}
 		\center{\includegraphics[width=1\linewidth]{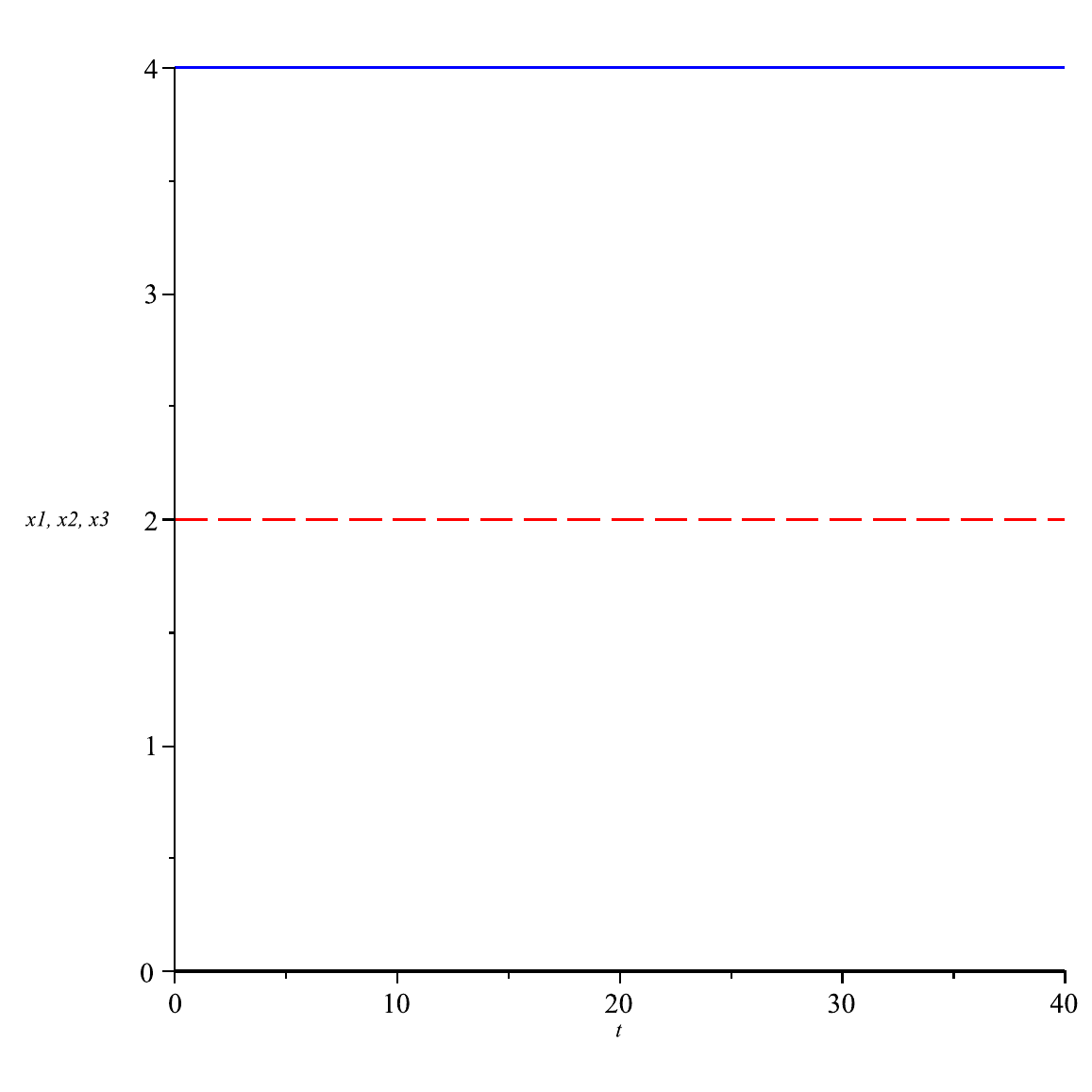}} a) \\
 	\end{minipage}
 	\hfill
 	\begin{minipage}[h]{0.4\linewidth}
 		\center{\includegraphics[width=1\linewidth]{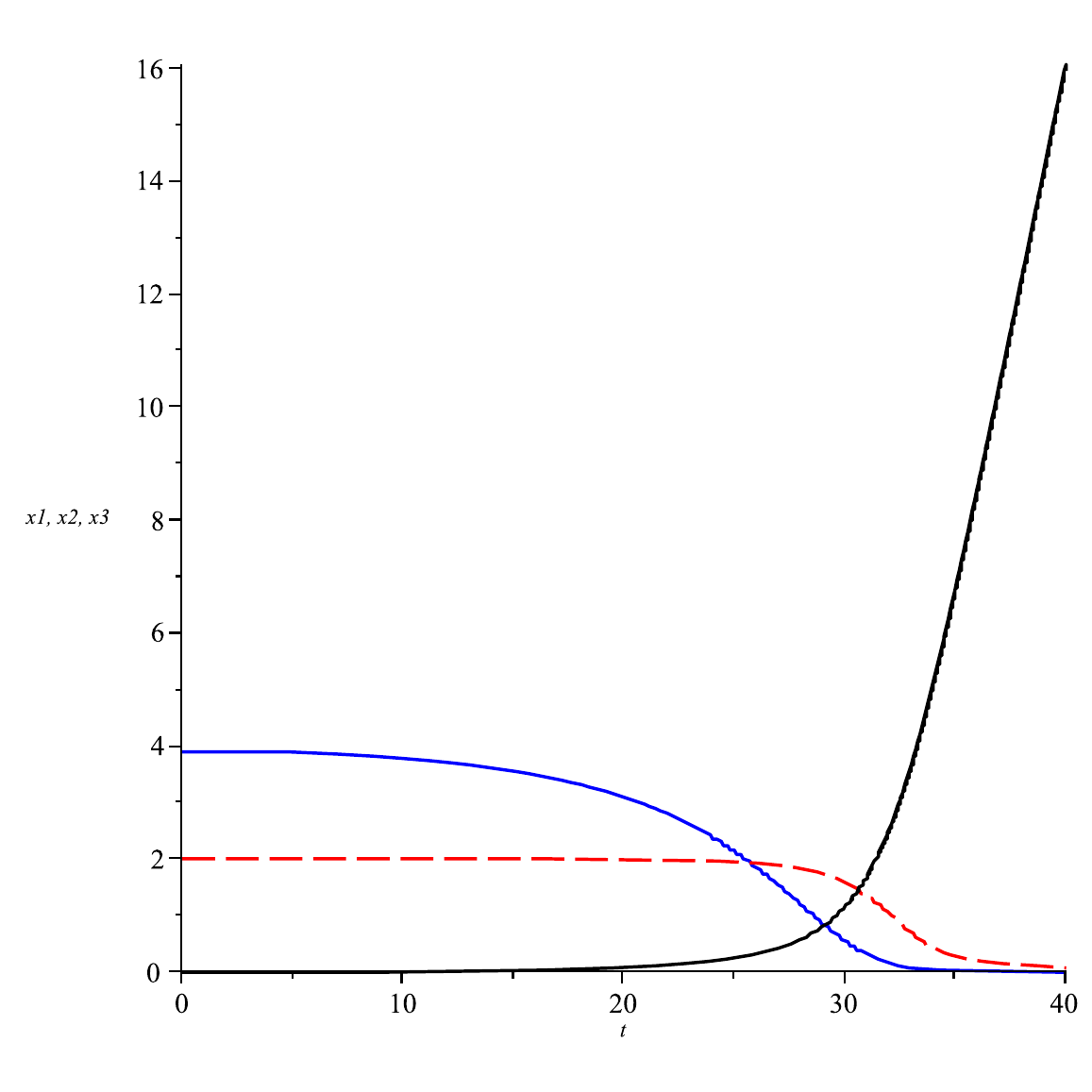}} \\b)
 	\end{minipage}
 	\vfill
 	\begin{minipage}[h]{1\linewidth}
 		\center{\includegraphics[width=0.4\linewidth]{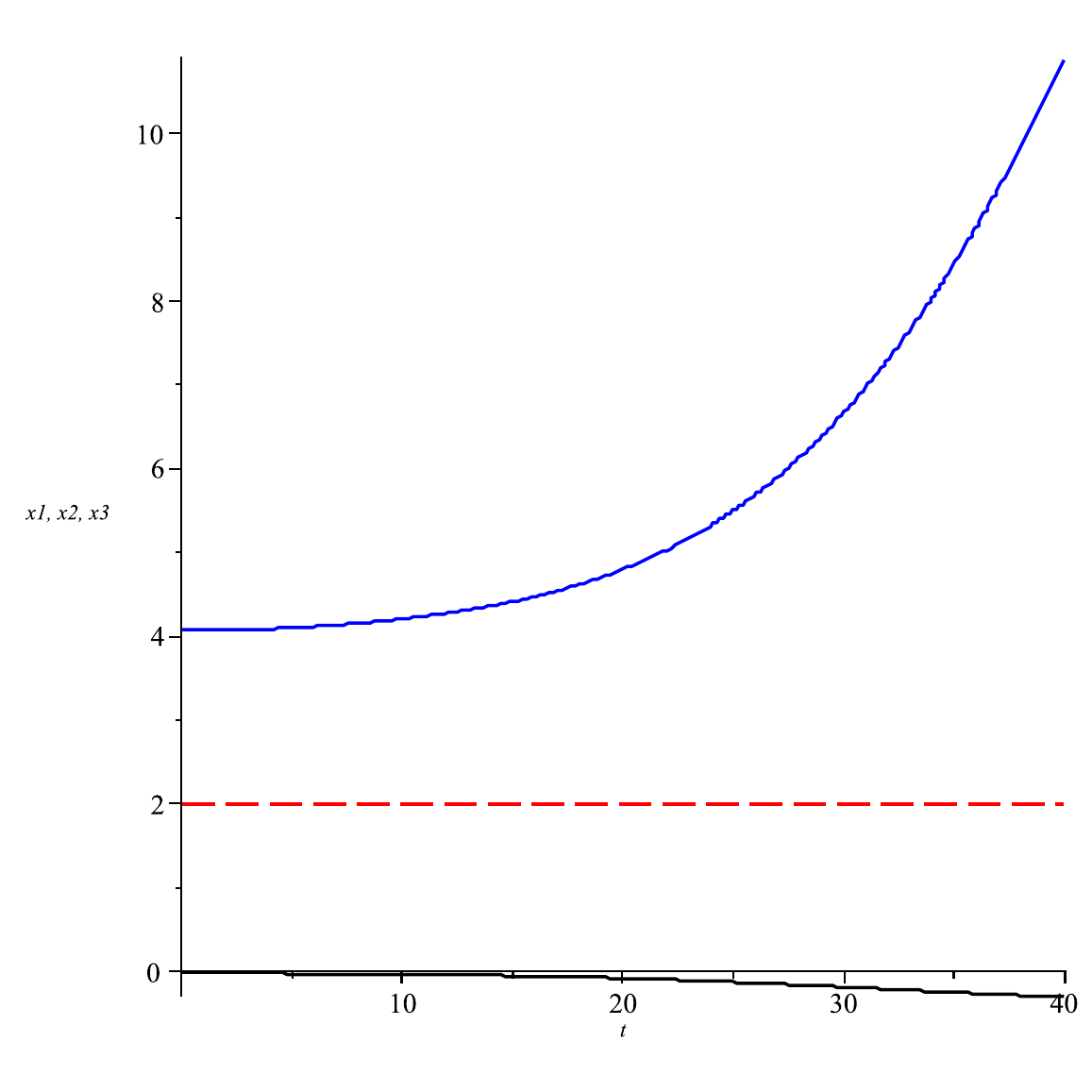}} c) \\
 	\end{minipage}
 	\vfill
 	\caption{Transient behavior of the closed-loop system \eqref{cloloo0} with different initial conditions: a) at the unstable equilibrium point $\chi^u(0)=(4,2,0)$; b) at $\chi^u(0)=(3.9,2,0)$; c) at $\chi^u(0)=(4.1,2,0)$. $\chi_1$--blue, $\chi_2$--dotted red and $\chi_3$--black.}
 	\label{fig2}
 \end{figure}    

\subsection{PI control with $d_1>0$: stability/instability of the equilibria}
\lab{subsec32}
%
In this subsection we do the (local) stability analysis of equilibrium when $d_1>0$ and present some simulations that illustrate the claims.
\subsubsection{Stability/Instability of the equilibrium}
\lab{subsubsec321}
%
\begpro
\lab{pro2}
Consider the system \eqref{sys} with $d_1>0$, verifying the equilibrium existence condition \eqref{equcon}, in closed-loop with the {\em PI control} \eqref{pi}. 
\begenu[{\bf C1}]
\item The closed-loop system has {\em two equilibrium} points, consistent with $\bar \chi_2=y_\star $, denoted 
\begalis{
\bar \chi^u:=(\bar \chi^u_1,y_\star ,\bar \chi^u_3) \in \rea^3\\
\bar \chi^s:=(\bar \chi^s_1,y_\star ,\bar \chi^s_3) \in \rea^3,
}
whose explicit expressions are given in \eqref{equequ} in the the proof below.
\item $\bar \chi^u$ is {\em unstable} for all choices of the tuning gains.
    
\item $\bar \chi^s$ is {\em asymptotically stable} provided the following inequalities hold:\footnote{As seen in the proof there are many other options of the various gains that will ensure asymptotic stability. The one given in {\eqref{inega}} is just one particular example, which is easy to satisfy becasue $d_1$ is, in general, a small number.}  
    \begin{subequations}
    \label{inega}
\begin{align}
     \label{inegap}
	K_P \geq {1 \over 2} d_1^2,  \\
    \label{inegai}
	 K_I \geq {5 \over 16}{ d_1 \over y^2_\star}.
\end{align}
\end{subequations}

\item  A domain of attraction associated to  $\bar \chi^s$ is given by 
$$
\cald:=\{\chi \in \rea^3\;|\;\chi^\top P \chi \leq \rho_\cald,\;\rho_\cald>0\},
$$
with $P \in \rea^{n \times n}$ a positive definite matrix satisfying the algebraic Lyapunov equation \eqref{lyaequ} for any positive definite matrix $Q \in \rea^{2 \times 2}$ and $\rho_\cald$ is a sufficiently small constant.
\endenu

\endpro
\begin{proof}
 The closed-loop dynamics takes the form
\begequ
\lab{cloloo}
\dot \chi=\begmat{-d_1 \chi_1+1-u_0 \chi_2 -K_I \chi_2 \chi_3-K_P\chi_2y_\star +K_P \chi^2_2\\ -d_2 \chi_2 + u_0 \chi_1+K_I \chi_1 \chi_3+K_P\chi_1 y_\star -K_P\chi_1\chi_2\\-\chi_2+y_\star  }=:f(\chi).
\endequ
Claim {\bf C1} follows from the definition of the assignable equilibrium set \eqref{assequ} leading to
\begali{
\lab{equequ}
\bar \chi^u &:=\begmat{{1 \over 2d_1}\Big(1 - \sqrt{1 -4d_1d_2y_\star ^2}\Big)\\ y_\star\\{1 \over K_I y_\star }\Big(1 -d_1 \bar \chi^u_1 - u_0 y_\star\Big)},\;\bar \chi^s :=\begmat{{1 \over 2d_1}\Big(1 + \sqrt{1 -4d_1d_2y_\star ^2}\Big)\\ y_\star\\{1 \over K_I y_\star }\Big(1 -d_1 \bar \chi^s_1 - u_0 y_\star\Big)},
}
with the associated constant control values
\begalis{
\bar u^u&:={-1 \over (\bar \chi^u_1)^2+y_\star ^2}[(d_1 - d_2)\bar \chi^u_1-1]y_\star \\
\bar u^s&:={-1 \over (\bar \chi^s_1)^2+y_\star ^2}[(d_1 - d_2)\bar \chi^s_1-1]y_\star .
}

Invoking Lyapunov's Indirect Method \cite[Theorem 4.7]{KHAbook}, we will prove the claims {\bf C2} and {\bf C3} of {Proposition} \ref{pro2}. Towards this end, we compute the Jacobian of the vector field $f(\chi)$ as
$$
\nabla f(\chi)=\begmat{-d_1 & -u_0 -K_I \chi_3 -K_Py_\star +2K_P \chi_2 & -K_I \chi_2 \\ u_0 + K_I \chi_3 + K_P y_\star - K_P \chi_2  & -d_2 -K_P \chi_1 & K_I \chi_1 \\ 0 & -1 & 0}.
$$
Evaluating the Jacobian at the equilibrium $\bar \chi=\col(\bar \chi_1,y_\star , \bar \chi_3)$ yields
$$
\nabla f(\bar \chi)=\begmat{-d_1 & -u_0  -K_I \bar \chi_3 + K_P y_\star  & -K_I y_\star  \\ u_0 + K_I \bar \chi_3& -d_2  -K_P \bar \chi_1& K_I  \bar \chi_1\\ 0 & -1 & 0}.
$$
The characteristic polynomial of this matrix is of the form $p(\lambda):=\lambda^3+a_2 \lambda^2+a_1 \lambda + a_0$, with 
\begalis{
& a_0:=- K_I(K_I \bar \chi_3 y_\star  - d_1 \bar \chi_1 + u_0 y_\star )\\
& a_1:= K_I^2\bar \chi_3^2+( 2u_0 \bar \chi_3 + \bar \chi_1-K_P \bar \chi_3 y_\star  )K_I + (d_1 \bar \chi_1 - u_0 y_\star )K_P + d_1d_2 + u_0^2\\
& a_2:= K_P \bar \chi_1 + d_1 + d_2.
}  

As indicated in the proof of Proposition \ref{pro1}, the inequalities \eqref{stacon} are the necessary and sufficient condition for (local) stability of the equilibria.  In the sequel we will verify these conditions for the two different equilibrium points of the closed-loop system.\\

\noindent {\bf Case of minimal current} The minimal current equilibrium is associated with the unstable equilibrium and is given as
\begequ
\lab{bchi1m}
\bar \chi_1^u:={1 \over d_1}(\hal - r),
\endequ
where for ease of future reference we have defined the constant
$$
r:={1 \over 2}\sqrt{1 -4d_1d_2y_\star ^2} > 0.
$$
The corresponding equilibrium for $\bar \chi^u_3$ is then
\begequ
\lab{bchi3m}
\bar \chi^u_3= {1 \over K_I y_\star }\Big(\hal + r - u_0 y_\star \Big).
\endequ
After some algebraic operations we can evaluate $a_0$ with \eqref{bchi1m} and \eqref{bchi3m} getting
$$
a_0=-2K_Ir,
$$
which is always negative, proving claim {\bf C2}.\\

\noindent {\bf Case of maximal current} We repeat here the same steps above for the maximal current equilibrium. The maximal current equilibrium and its associated component for $\bar \chi_3^s$ are given as\footnote{Notice that the minimal and maximal current equilibria differ only on the sign of the constant $r$.} 
\begali{
\nonumber
\bar \chi_1^s&:={1 \over d_1}(\hal + r)\\
\lab{chis}
\bar \chi^s_3 &= {1 \over K_I y_\star }\Big(\hal - r - u_0 y_\star \Big).
}
Some lenghty, but straightforward calculations with these two equations yields
\begalis{ 
a_0 &=2K_Ir\\
a_1&={1 \over 4d_1 y_\star ^2}\Big(4 d_1^2 d_2 y_\star ^2+8rK_Py_\star ^2d_1+4\Big(r-\hal\Big)^2d_1+4K_Iy_\star ^2(\hal+r)   \Big)\\        
a_2&={K_P \over d_1}\Big(\hal +r\Big)+d_1+d_2.
}
which are all positive. It remains only to prove that 
\begequ
\lab{a1a2mina0}
a_1a_2-a_0>0.
\endequ
In Appendix \ref{appa} we present the {\tt Maple} calculations that compute \eqref{a1a2mina0} multiplied by a positive factor, denoted $f_1$. From these expressions we identify three polynomials that recover all the negative terms appearing in $f_1$, namley: 
\begalis{
p_1(r)&:=d_1^3(8r^2-8r+2) \\
p_2(r)&:=d_2d_1^2(8r^2-8r+2)+d_1^2(-8K_I y_{\star}^2r+4K_I y_{\star}^2)\\
p_3(r)&:=8K_Pd_1r(r^2-{1 \over 2}r - {1 \over 4})
}

Applying to these polynomials the following equivalence
$$
p(r):=r^2+a r +b \geq 0 \quad \Longleftrightarrow \quad 4b \geq a^2,
$$
we conclude that $p_1(r)$ is always nonegative. Regarding $p_2(r)$, note that its first term on the right-hand side is clearly nonegative,  and its second term is larger equal to $-4K_Iy_{\star}^2d_1^2r^2$, which yields
\begin{align}
p_2(r) &\geq d_2d_1^2(8r^2-8r+2)+4K_Iy_{\star}^2d_1^2(r-1)^2-4K_Iy_{\star}^2d_1^2 r^2 \nonumber \\
&\geq -4K_Iy_{\star}^2d_1^2 r^2 \nonumber.
\end{align}
To dominate this term we take the term $8K_PK_Iy_\star ^2r^2$, which appears in the last right hand side product of $f_1$, thus ensuring the inequality \eqref{inegap} holds. 

Regarding $p_3(r)$ it is easy to show that the second order polynomial is larger equal to $-{5 \over 16}$, hence we have
$$
p_3(r) \geq -{5 \over 2} K_Pd_1r.
$$
To dominate this term we take the term $8K_PK_Iy_\star ^2r$, which appears in the last right hand side product of $f_1$, provided that the inequality \eqref{inegai} holds. Hence, the claim {\bf C3} is proved.

To prove the claim {\bf C4} we define the Lyapunov equation
\begequ
\lab{lyaequ}
A^\top P + P A =-Q,
\endequ
where we defined the matrix 
$$
A:=\nabla f(\chi^^s),
$$
and recall that the positive definite function $V(\delta \chi):=\delta \chi^\top P \delta \chi$ qualifies as a Lypunov function for the linearized system 
$$
{\stackrel{\dot {\frown}} {\delta \chi}}=A \delta \chi.
$$
\end{proof}
\subsubsection{Simulation results}
\lab{subsubsec322}
%
In Figs. {\ref{fig3} and \ref{fig4}} we present simulations of the closed-loop system \eqref{cloloo} with different initial conditions for the following parameters:
$$
y_\star =1,d_1= {1 \over 4},d_2={3 \over 4},K_I=1,K_P=2,u_0=\hal.
$$
First, we set the initial conditions at (and close) the unstable equilibrium $\bar \chi^u$ associated to the minimum current. Then, we tune the gains to ensure the equilibrium associated to the maximal current is stable verifying \eqref{a1a2mina0} and, again, set the initial conditions at (and close) the stable equilibrium $\bar \chi^s$. 

  \begin{figure}
 	\begin{minipage}[h]{0.4\linewidth}
 		\center{\includegraphics[width=1\linewidth]{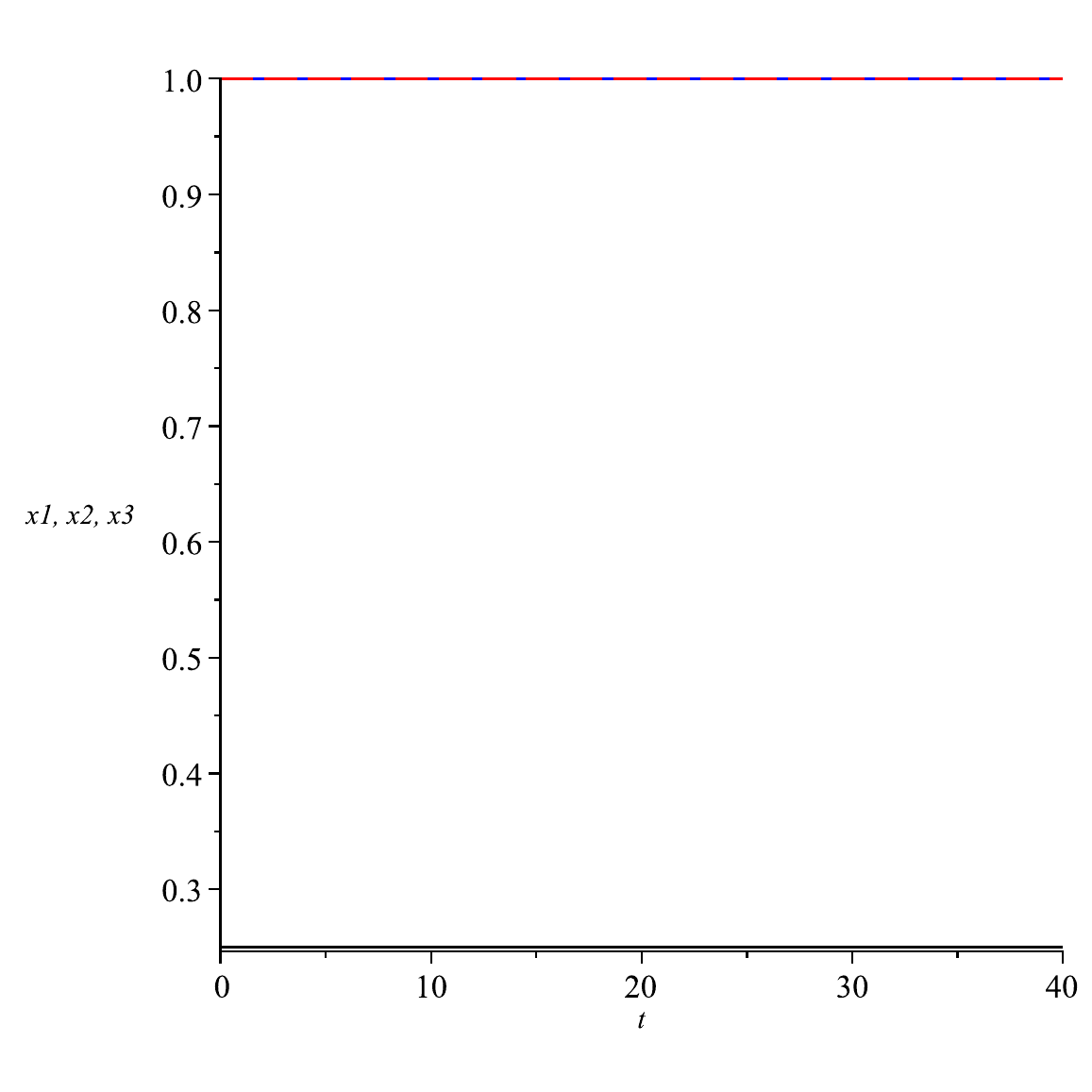}} a) \\
 	\end{minipage}
 	\hfill
 	\begin{minipage}[h]{0.4\linewidth}
 		\center{\includegraphics[width=1\linewidth]{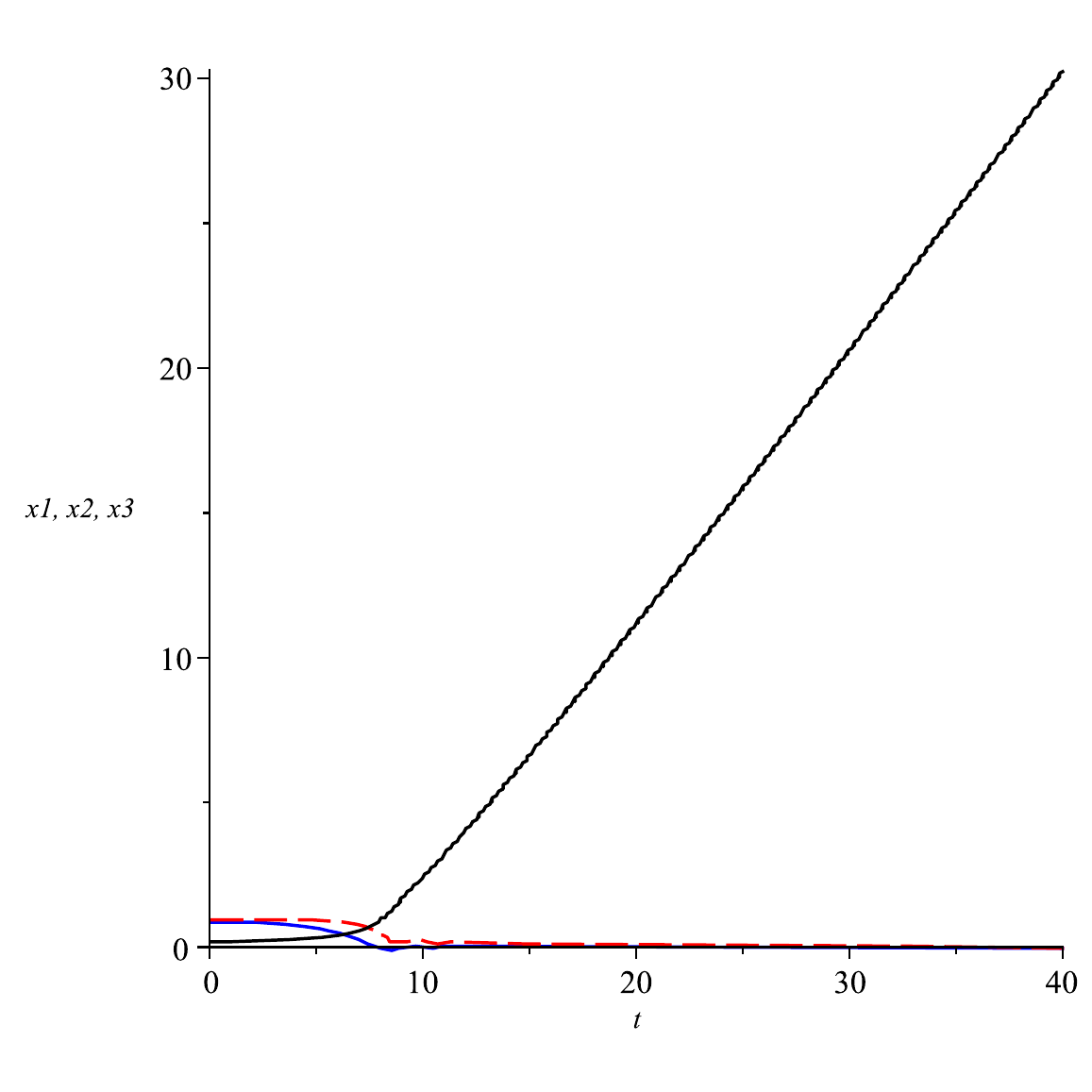}} \\b)
 	\end{minipage}
 	\vfill
 	\begin{minipage}[h]{1\linewidth}
 		\center{\includegraphics[width=0.4\linewidth]{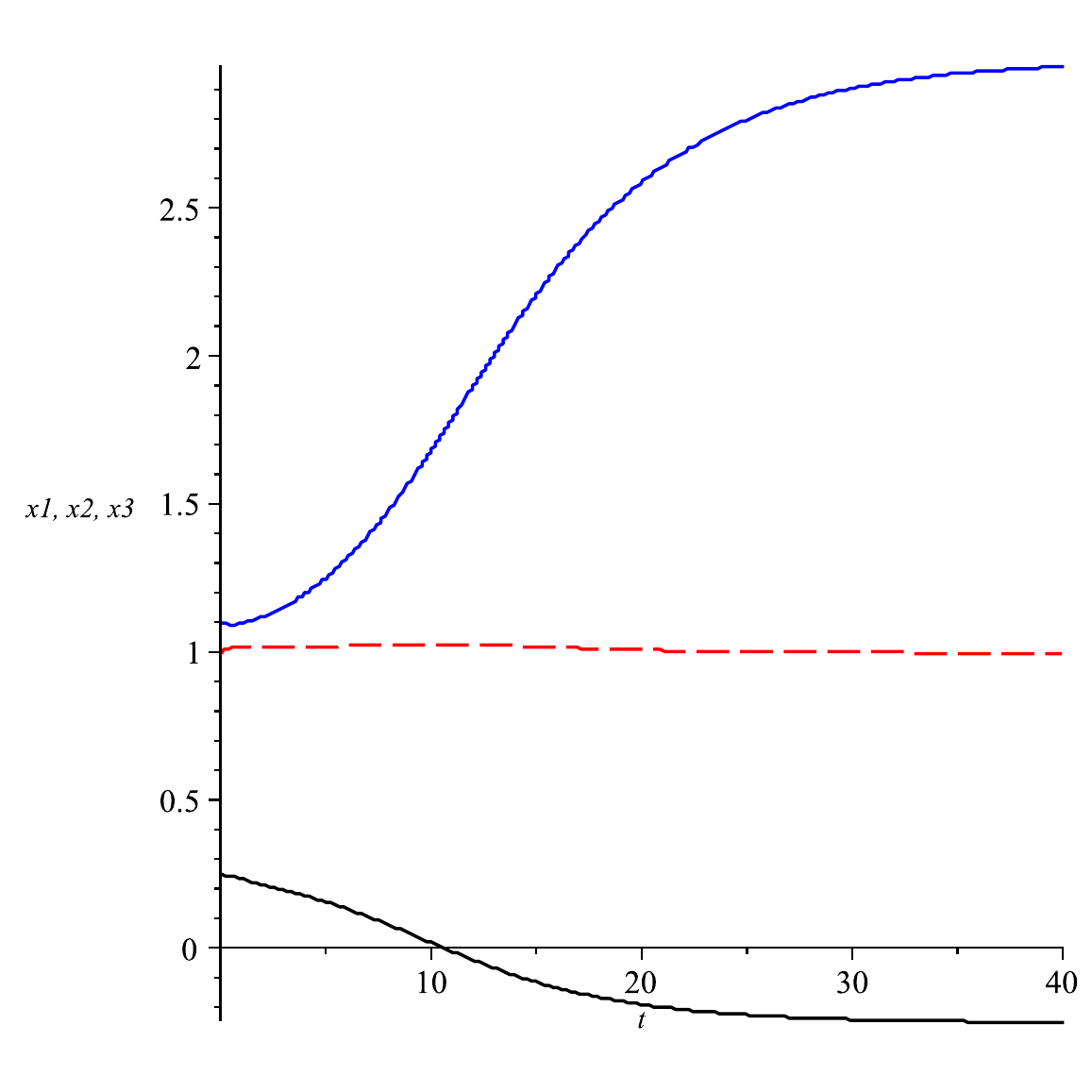}} c) \\
 	\end{minipage}
 	\vfill
 	\caption{Transient behavior of the closed-loop system \eqref{cloloo} with different initial conditions: a) at the unstable equilibrium point $\chi^u(0)=(1,1,{1 \over 4})$; b) at $\chi^u(0)=(0.9,1,{1 \over 4})$; c) at $\chi^u(0)=(1.1,1,{1 \over 4})$.  $\chi_1$--blue, $\chi_2$--dotted red and $\chi_3$--black.}
 	\label{fig3}
 \end{figure}
 
  \begin{figure}
 	\begin{minipage}[h]{0.4\linewidth}
 		\center{\includegraphics[width=1\linewidth]{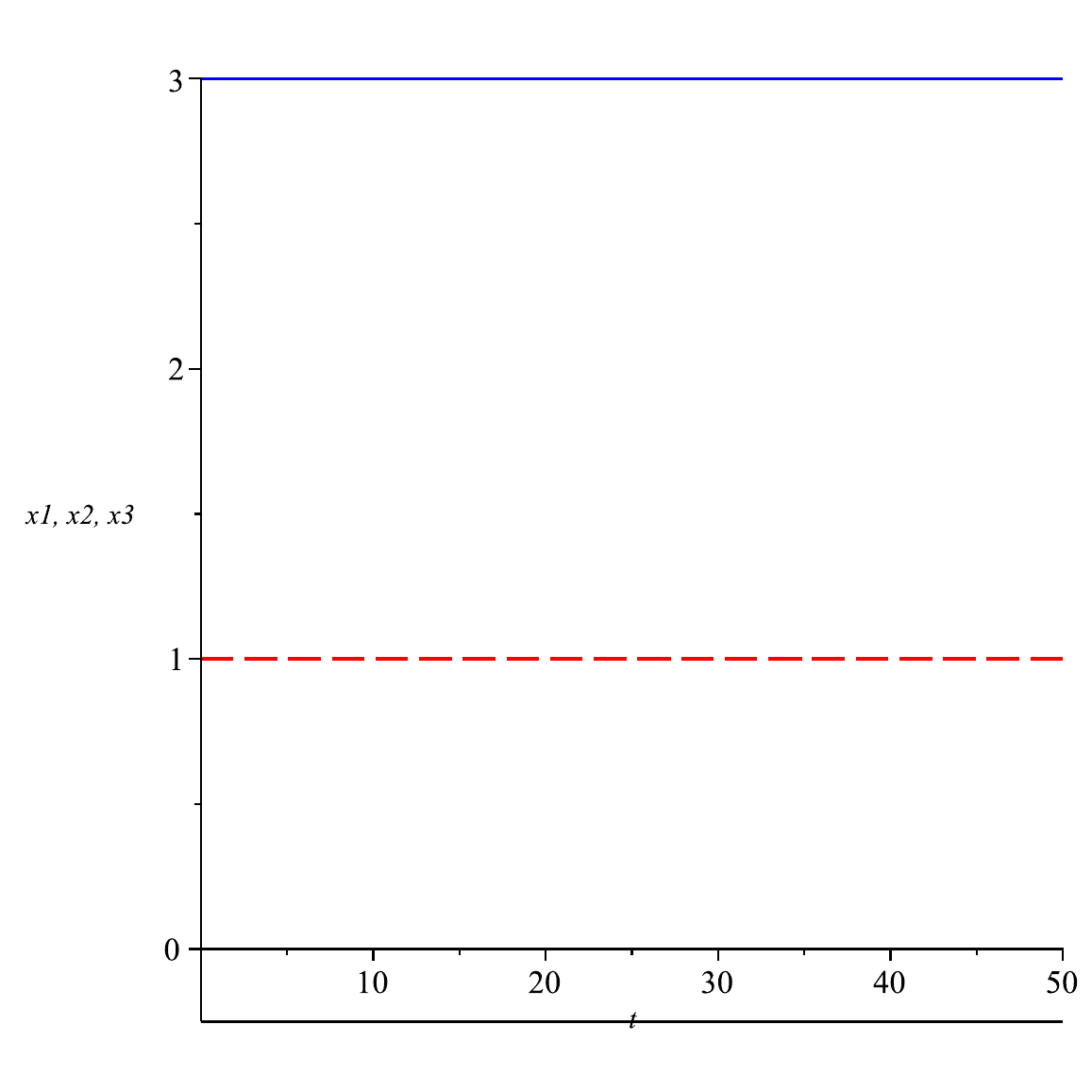}} a) \\
 	\end{minipage}
 	\hfill
 	\begin{minipage}[h]{0.4\linewidth}
 		\center{\includegraphics[width=1\linewidth]{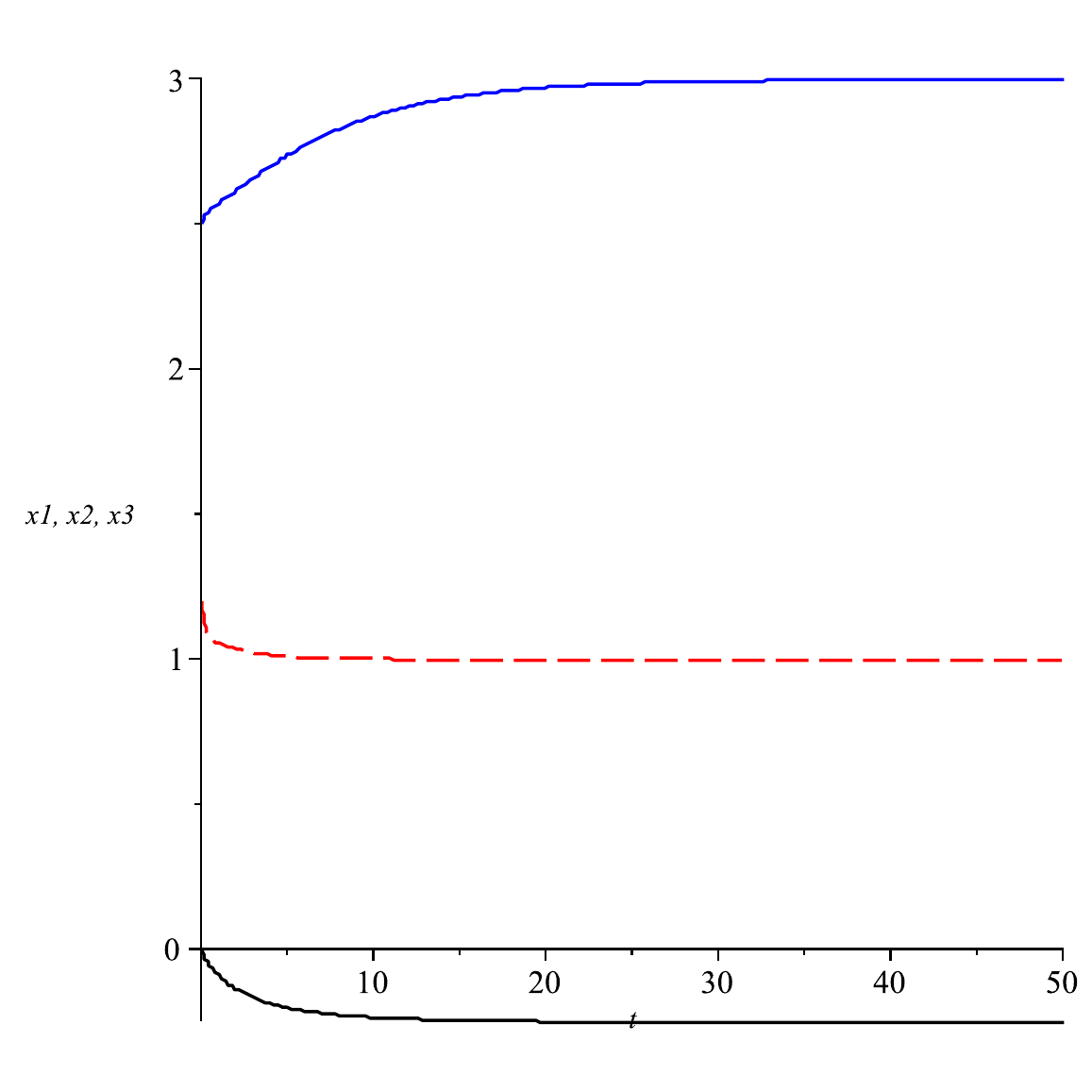}} \\b)
 	\end{minipage}
 	\vfill
 	\begin{minipage}[h]{1\linewidth}
 		\center{\includegraphics[width=0.4\linewidth]{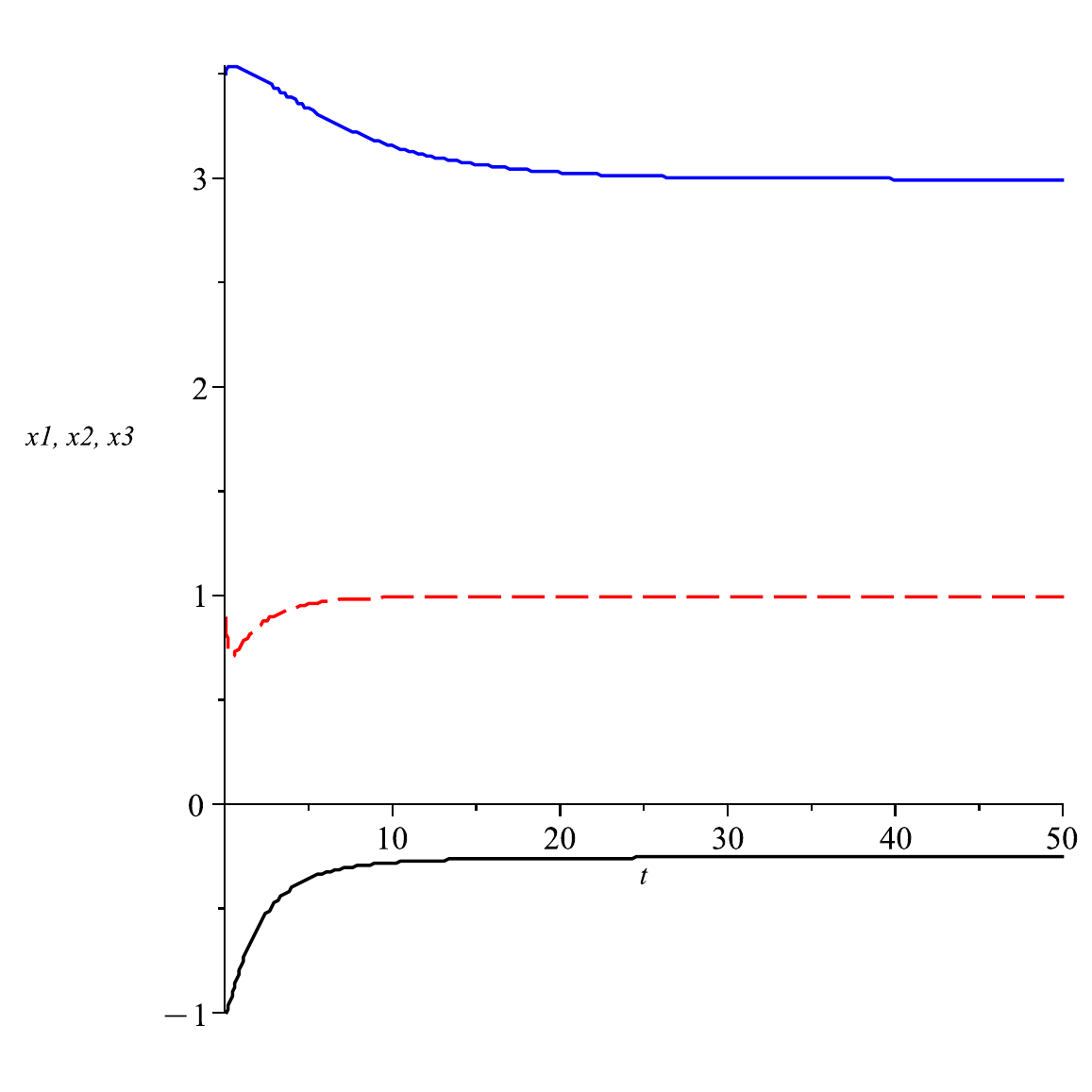}} c) \\
 	\end{minipage}
 	\vfill
 	\caption{Transient behavior of the closed-loop system \eqref{cloloo} verifying \eqref{a1a2mina0} with different initial conditions: a) at the stable equilibrium point $\chi^s(0)=(3,1,-{1 \over 4})$; b) at $\chi^s(0)=(2.5,1.2,0)$; c) at $\chi^s(0)=(3.5,0.9,-1)$. $\chi_1$--blue, $\chi_2$--dotted red and $\chi_3$--black.}
 	\label{fig4}
 \end{figure}
As shown in Figs. \ref{fig3} and \ref{fig4}, the instability of $\chi^u$ and the stability of $\chi^s$ are evident, thereby corroborating Proposition 2.

Similarly to the case $d_1=0$, the surface 
$$
\cals_u:=\{(\chi_1,\chi_2)\;|\;\chi_1-d_1 \chi_1^2 -d_2 \chi_2^2=0\},
$$
which clearly contains the assignable equilibrium point given in \eqref{assequ}, defines a separatrix with the same  properties as $\cals_0$. Again,  this surface precisely coincides with the place where the derivative of the systems energy $H(x)$ equal zero.
%
\section{Voltage Feedback Passivity-based Control}
\lab{sec4}
%
In this section we recall three voltage feedback controllers which ensure global regulation of the voltage signal. All three of them are designed following the approach of PBC \cite{ORTetalcsm}, which has had enormous success in applications. The first two schemes are nonlinear voltage-feedback static controllers, and have been reported in  \cite{RODetal} and \cite{ZHAetal}, respectively. They both belong to the class of Interconnection and Damping Assignment (IDA)-PBC first reported in \cite{ORTetalaut02,ORTGAR}. The third one is a PID-PBC, first reported in \cite{HERetal} and further developed in \cite{BOBetal,ORTetalbookpid}.
\subsection{Two nonlinear output feedback controllers}
\lab{subsec41}
%
We give below, without proof, the propositions pertaining to the two static controllers reported in  \cite{RODetal} and \cite{ZHAetal}, respectively. The proofs may be found in the corresponding reference.

\begpro \em
\lab{pro3}
 \cite{RODetal} Consider the Boost converter model \eqref{sys} with $d_1=0$ in closed-loop with the IDA-PBC
\begequ
\lab{conhug}
u={1 \over y_\star }\Big({x_2 \over y_\star }\Big)^\alpha,
\endequ
where $\alpha \in (0,1)$ is a tuning gain. Then, the equilibrium point $x_\star :=(d_2 y_\star ^2,y_\star )$ is {\em asymptotically stable}.
\endpro

\begpro \em
\lab{pro4}
\cite{ZHAetal} Consider the Boost converter model \eqref{sys} with $d_1=0$ in closed-loop with the IDA-PBC
\begequ
\lab{conmen}
u=k{x_2 \over x_2^2+(k-1)y_\star ^2},
\endequ
where {$k>3$} is a tuning gain. 
\begenu[{\bf P1}]
\item The equilibrium point $x_\star :=(d_2 y_\star ^2,y_\star )$ is {\em asymptotically stable}.
\item There exists a domain of attraction, defined by a Lyapunov function\footnote{The expression of $P(x)$ is given in \cite{ZHAetal} of the closed-loop dynamics $P(x)$}, that guarantees that if we start in this set, the state $x(t)$ remains in the positive quadrant. 
\endenu
\endpro

\subsection*{Simulation results}
%
In {Figs. \ref{fig5} and \ref{fig6} below} we show some simulation results of the two controllers above.  As seen in the figures both controllers ensure asymptotic stability of the desired equilibrium. Some other properties, pertaining to the performance of the controllers, {\em e.g.}, domain of attraction, speed of convergence, etc, may be found in \cite{RODetal} and \cite{ZHAetal}.

 \begin{figure}
 	\begin{minipage}[h]{0.5\linewidth}
 		\center{\includegraphics[width=1.1\linewidth]{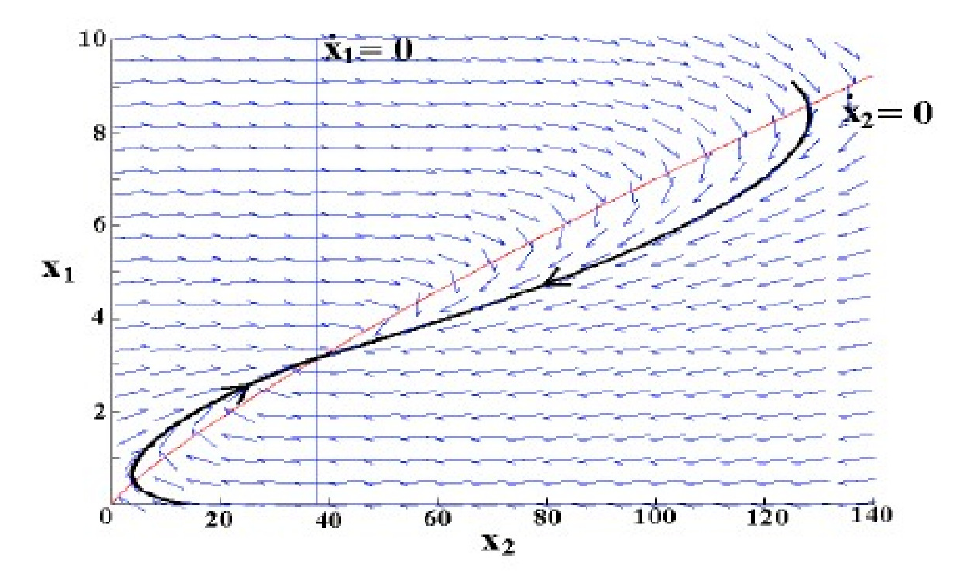}} a) \\
 	\end{minipage}
 	\hfill
 	\begin{minipage}[h]{0.5\linewidth}
 		\center{\includegraphics[width=1\linewidth]{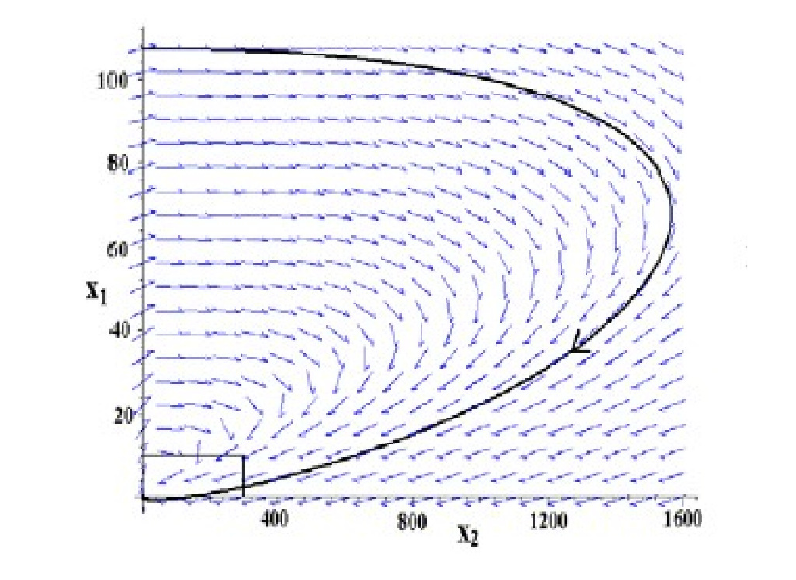}} \\b)
 	\end{minipage}
 	\vfill
 	\caption{a) Phase plane of the converter model in closed-loop with the control \eqref{conhug}, showing one trajectory (in black); b) an enlarged view illustrating the large size of the domain of attraction (see the scales).}
 	\label{fig5}
 \end{figure}
 
 \begin{figure}
 	\begin{minipage}[h]{0.45\linewidth}
 		\center{\includegraphics[width=1.1\linewidth]{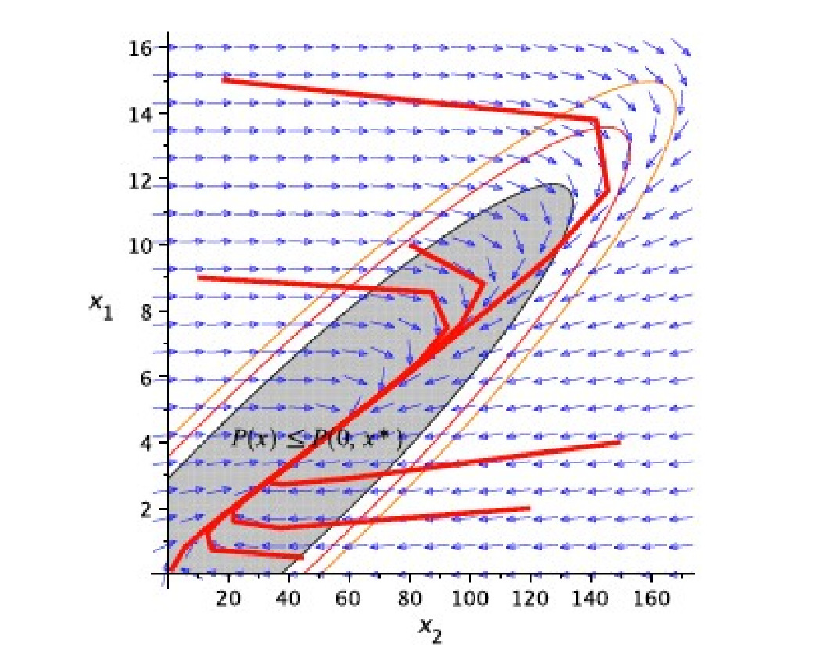}} a) \\
 	\end{minipage}
 	\hfill
 	\begin{minipage}[h]{0.45\linewidth}
 		\center{\includegraphics[width=1.1\linewidth]{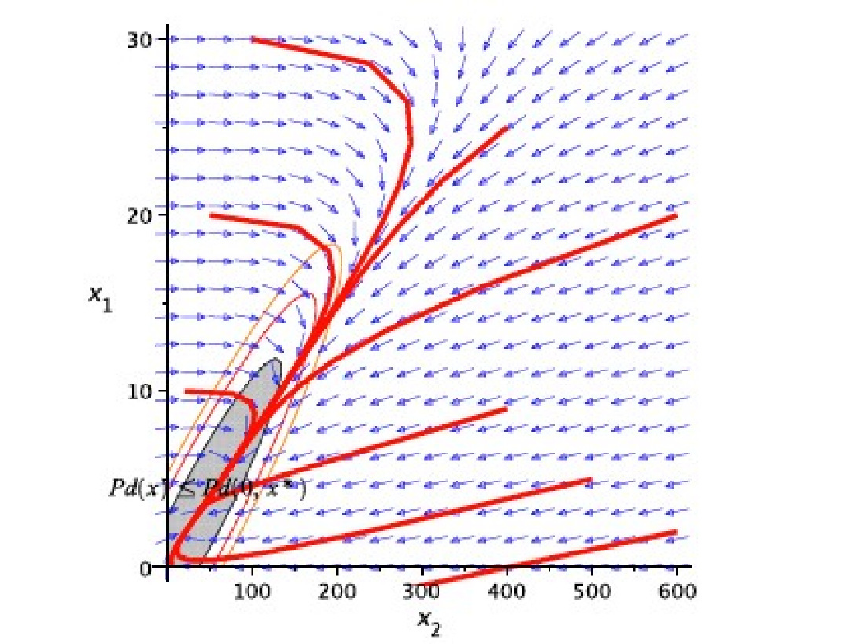}} \\b)
 	\end{minipage}
 	\vfill
 	\caption{a) Phase plane of the converter model in closed-loop with the control \eqref{conmen}, showing some trajectories (in red) and  an estimate of the domain of attraction that---as indicated by the arrows---ensures $x(t) >0$ (in gray); b) An enlarged version. }
 	\label{fig6}
 \end{figure}
 
\begrem
\lab{rem7}
The simplicity of both controllers can hardly be overestimated. Consisting only of some very simple algebraic operations. Also, it should be underscored that the neither one of the controllers require {\em any prior knowledge} about the converter parameters, making them robust with respect to parameter uncertainty. 
\endrem

\begrem
\lab{rem8}
One drawback of the two propositions above is that it is assumed that the parasitic resistance $d_1$ is equal to zero. Current research is under way to extend these results for the case $d_1 \neq 0$. 
\endrem
\subsection{A PID passivity-based controller}
\lab{subsec42}
%
We finish this section presenting the PID-PBC reported in  \cite{BOBetal}. This scheme is an output-feedback version of the original PID-PBC first reported in \cite{HERetal} and further developed in \cite{ORTetalbookpid}. The main contribution of \cite{BOBetal} is that, implementing an adaptive observer, removed the need for the {\em full state} feedback required in \cite{HERetal}. Another interesting feature of the PID-PBC is that it is applicable to a large class of power converters, in contrast to the schemes in \cite{RODetal} and \cite{ZHAetal}, which are developed for the Boost converter only---moreover, for the Boost converter the PID-PBC design does not need to assume $d_1=0$. The prize that is paid for these generalizations is in the {\em complexity} of the implementation that, as will be seen below, is several orders of magnitude bigger than the two static controllers of Propositions \ref{pro3} and \ref{pro4}.

As explained above, the controller in \cite{BOBetal} consists of a state observer and a PID-PBC of the form given in \cite{HERetal,ORTetalbookpid}. To simplify the presentation we split the main stability result in two parts: first, the design of the state observer for the Boost converter model \eqref{sys}. Second, the implementation of a {\em certainty equivalent} version of the PID-PBC of \cite{HERetal,ORTetalbookpid}, where the actual converter current is replaced by its estimate generated by the observer. It should be underscored that the observer has {\em finite convergence time} (FCT) properties, hence the good behavior of the certainty equivalent PID-PBC is guaranteed. 

In the proposition below we propose a state observer that ensures FCT of the estimated state, under the following extremely  weak  interval excitation assumption \cite{KRERIE} on a scalar signal $\Delta$, which is defined below.
\begin{assumption}
\label{ass1}
{Fix a small constant $\mu \in (0, 1)$ and a positive number $\gamma$. There is a time $t_c \in (0,\infty)$ such that}
\begequ
\lab{sufexccon}
\int^{t_c}_0 \Delta^2(\tau)d\tau\geq-\frac{1}{\gamma}\ln(1-\mu).
\endequ
\end{assumption}

\begpro 
\lab{pro5}
\cite{BOBetal} Consider the Boost converter model \eqref{phsys} with measurable output the voltage
$$
y=C x = [0\;1]x=x_2.
$$
Define the observer via
\begsubequ
\lab{gpebo}
\begali{
\dot \xi & =F(u)\xi+\begmat{1 \\ 0} \label{dotxi}\\
\dot \Phi &=F(u)\Phi,\;\Phi(0)=I_2 \lab{dotphi}\\
\dot Y&=-\lambda Y + \lambda \Phi^\top  C^\top (C\xi-y) \lab{doty}\\
\dot \Omega&=-\lambda \Omega + \lambda \Phi^\top  C^\top  C\Phi \lab{dotome}\\
\dot\omega &=-\gamma\Delta^2\omega,\; \omega(0)=1 \lab{dotw}\\
\dot {\hat \theta}&=\gamma\Delta(\mathcal{Y}-\Delta\hat \theta),\label{dothatthe}
}
\endsubequ
with $F(u)$ given in \eqref{f}, $\lambda>0,\; \gamma>0$, free tuning parameters and
\begali{
\mathcal{Y}:=\text{adj}\{\Omega\}Y,\;\Delta :=\det\{\Omega\},
\lab{calydel}
}
where $\adj\{\cdot\}$ is the adjunct  matrix. The state estimate
\begequ
\lab{hatxfct}
\hat x=\xi+\Phi\hat\theta_{\tt FCT}
\endequ
where we introduced the vector
\begequ
\lab{hatthefct}
\hat\theta_{\tt FCT}:=\frac{1}{1-\omega_c}\left[\hat \theta-\omega_c\hat\theta(0)\right],
\endequ
and $\omega_c$ is defined via the clipping function
\begin{align}\nonumber
\omega_c=
\begin{cases}
\omega~~~~~~~\text{if}~~\omega \leq 1-\mu,\\
1-\mu~~\text{if}~~\omega > 1-\mu.
\end{cases}
\end{align}
ensures that, for all initial conditions $(x(0),\xi(0),Y(0),\Omega(0),\hat \theta(0)) \in  \rea^2 \times \rea^{2} \times \rea^{2} \times \rea^{2 \times 2} \times \rea^{2}$, we have that
\begequ
\lab{hatequx}
\hat x(t)=x(t),~~\forall t>t_c
\endequ
with all signals bounded provided $\Delta$ verifies Assumption \ref{ass1}.
\end{proposition}

As thoroughly discussed in  \cite{HERetal,ORTetalbookpid} the key step in the implementation of the PI-PBC is the generation of an output signal that is {\em shifted passive} \cite[Section 4.2.1]{ORTetalbookpid}. This property combined with the well-known passivity of the PI ensures stability of the closed-loop. Equipped with the FCT state observer of Proposition \ref{pro5} we can implement a (certainty equivalent) voltage-feedback PID-PBC as follows.

\begpro \em
\lab{pro6}
\cite{BOBetal} Consider the Boost converter model \eqref{phsys} in closed-loop with the PID-PBC
\begali{
\nonumber
\dot x_c&= y_{PI}\\
u&= -K_P y_{PI}-K_I x_c,
\lab{pidpbc}
}
where 
$$
y_{PI}=x_1^\star  x_2 - y_\star \hat x_1,
$$
with $\hat x_1(t) \in \rea$ the estimate of $x_1$ generated in Proposition \ref{pro5}. Then, the equilibrium point $(\bar x_1,\bar x_2,\bar x_c):=(d_2 y_\star^2,y_\star,0)$ is {\em asymptotically stable}.
\endpro
%
\section{Concluding Remarks and Future Research}
\lab{sec5}
%
We have carried out a detailed analysis of the behavior of the Boost converter under the action of {\em voltage output} controllers. We have considered the classical PI control and three different classes of PBC. We have proven the important role of the inductor parasitic resistance ($R$ in \eqref{boost}) showing, in particular, that the PI controller is {\em always} unstable if $R=0$. In the alternative case that $R >0$ we proved that the dynamics of the PI control exhibits two equilibria, one of them {\em always} unstable and the other one stable for {\em sufficiently large} PI gains. But it is important to note that the equilibrium of practical interest (the one with the lowest inductor current) corresponds to the unstable case.

Regarding the PBC controllers, we presented two very simple nonlinear static voltage feedback controllers, exhibiting {\em excellent} dynamic performance. Both schemes are restricted to the particular case when $R=0$. Current research is under way to extend the results of the PBC to consider $R>0$. Two additional research lines we are pursuing for the PBCs are, on one hand, to extend their application to other DC-to-DC converter topologies, {\em e.g.}, Buck-Boost or \v{C}uk and to consider other type of loads, for instance, {\em constant power loads or a combination of resistive and constant power loads.} Finally, we expect to be able to carry out experimental work of the proposed controller topologies in the near future.      
%

\appendix
\section{{\tt Maple} Computation of the Term \aTerm}
\lab{appa}
\vspace{1em} 

\begin{center}
\includegraphics[scale=1]{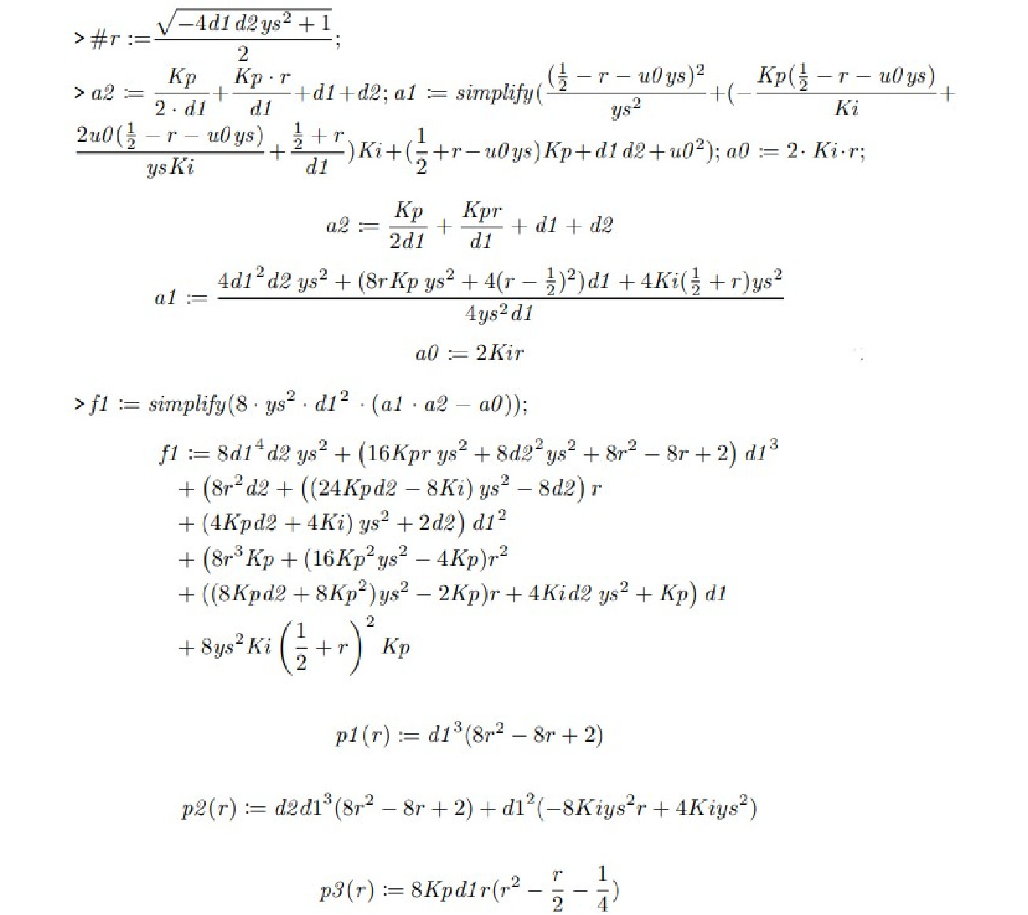}
\end{center}
\end{document}